\documentclass[letterpaper,12pt]{article}
\usepackage{graphicx}
\usepackage{amssymb,amsmath,amsfonts}
\usepackage[affil-it]{authblk} 
\usepackage[sort,compress,numbers,super]{natbib}

\usepackage{threeparttablex}
\usepackage{multirow}
\usepackage{latexsym}
\usepackage{mathrsfs} 
\usepackage{hyperref}\hypersetup{colorlinks=true,linkcolor=blue,citecolor=blue}

\usepackage{enumerate}
\usepackage{multirow}

\usepackage[dvipsnames, table]{xcolor}

\usepackage[margin=1in,headsep=1em]{geometry}
\usepackage[font=small]{caption}

\usepackage{titlesec}
\titleformat{\section}
{\bf\filright\large}
{\thesection.}
{0.5em}
{\bf\large}

\titleformat{\subsection}
{\bf\small\filright}
{\thesubsection}
{0.5em}
{\bf\small}

\emergencystretch=\maxdimen
\newcommand{\gi}[1]{{^*#1}}
\newcommand{\tens}[1]{\mathsf{#1}}
\renewcommand{\vec}[1]{\mathbf{#1}}

\usepackage{amsthm}
\newtheorem{theorem}{Theorem}
\newtheorem{definition}[theorem]{Definition}
\newtheorem{corollary}[theorem]{Corollary}
\newtheorem{proposition}[theorem]{Proposition}
\theoremstyle{remark}
\newtheorem{remark}[theorem]{Remark}

\begin{document}

%
%
%
%
%
%
%

\title{\LARGE A Note On Galilean Invariants In Semi-Relativistic Electromagnetism}
\author{Yintao Song\thanks{Email: \href{mailto:ytsong@umn.edu}{ytsong@umn.edu}.}\\
\normalsize\it Department of Aerospace Engineering and Mechanics, University of Minnesota, \\
\normalsize\it Minneapolis, Minnesota 55455, USA}
\date{\vspace{-2em}}
\maketitle

\begin{abstract}
The incompatibility between the Lorentz invariance of classical electromagnetism and the Galilean invariance of continuum mechanics is one of the major barriers to prevent two theories from merging. In this note, a systematic approach of obtaining Galilean invariant field variables and equations of electromagnetism 
within the semi-relativistic limit is reviewed and extended. 
In particular, the Galilean invariant forms of Poynting's theorem and the momentum identity, two most important electromagnetic identities in the thermomechanical theory of continua, are presented. In this note, we also introduce two frequently used stronger limits, namely the magnetic and the electric limit.
The reduction of Galilean invariant variables and equations within these stronger limits are discussed.
\end{abstract}

\section{Introduction}
Recently, the rapidly growth of many technological applications involving both mechanical and electromagnetic properties of materials,
such as MEMS, elastic dielectrics, and piezoelectric materials, 
stimulates strongly the interests in the field theory of thermomechanical continua interacting with classical electromagnetism \cite{Hutter, Kovetz, Fosdick, Ericksen2005, Ericksen2007, Ste2009, Liu2013}. 
The subject was initiated by the theories of elastic dielectric \cite{Toupin} and magnetoelasticity \cite{Brown1966}.
Later, many efforts \citep{Pen1967, Truesdell_Toupin, Maugin1988, Tiersten, Hutter, Kovetz, Ericksen2007, Fosdick, Ste2009, Liu2013} have been made to a fully unified theory concerning electricity, magnetism and mechanics all together.

However, until now the basic formulation of the integrated theory is still lack of a universally accepted version. One of the major barriers to the fusion of electromagnetism and thermomechanics of continua, as pointed out by \citet{Hutter, Fosdick} and others, is the complexity of addressing the issue of space-time invariance. This issue becomes particularly significant when the velocity field of particles in materials cannot be neglected. Continuum mechanics is required to be invariant (covariant) under Galilean transformations of the three dimensional Euclidean space, while classical electromagnetism is Lorentz invariant in the four dimensional Minkowski spacetime. The physical variables, such as electric field, magnetization and the Lorentz force, directly adopted from electromagnetism are not Galilean invariant. Consequently, physical laws and constitutive relations in terms of these variables will have different forms for observers doing measurements in different Galilean inertial frames. 

A popular strategy of addressing this issue is to introduce the Galilean invariant forms of various field variables in electromagnetism.
Such Galilean invariant forms (Galilean invariants, partial potentials) have been discussed in detail in the books of \citet{Hutter} and \citet{Kovetz} among others.
Other authors, such as \citet{Tiersten}, \citet{Fosdick}, have also carefully studied the transformation rules of field variables in electromagnetism
under Galilean transformations of the Euclidean space, without explicitly mentioning the notion of Galilean invariant forms.
It is useful to have a clear summary of these Galilean invariant forms, because, using \citet{Truesdell_Toupin}'s words,
\begin{quote}
``In most elementary and even advanced texts on electromagnetic theory, a clear distinction is not made between the partial
potentials $\mathfrak D$ and $\mathfrak H$ and the resultant potential $\mathbf D$ and $\mathbf H$ in the discussion of
polarizable and magnetizable media.''
\end{quote}
The present note aims to give a systematic review and some extensions of the Galilean invariant forms of field variables and equations in classical electromagnetism
within the semi-relativistic limit (defined later).

Once the Galilean invariants are determined, all the equations in classical electromagnetism can be rewritten in terms of 
them rather than the original non-objective field variables, as we will shown later in this note and also can be found in the literature
\citep{Pen1967, Truesdell_Toupin, Maugin1988, Tiersten, Hutter, Kovetz, Fosdick, Ericksen2005, Ericksen2007, Ste2009}. 
What's more, this enables us to use \emph{only} Galilean invariant variables in constitutive relations.
For example, a dielectric material often has the constitutive relation like $\vec P = \chi\vec E$, where $\vec P$ is the polarization
and $\vec E$ is the electric field strength. $\chi$ is a material constant to be determined by experimental characterization.
It has been noticed since the time of \citet{Lorentz} that this constitutive relation is not Galilean invariant.
Thus, some researchers, for example \citet{Landau_Continuous, Truesdell_Toupin}, have postulated an alternative way of writing this constitutive relation that is $\vec P = \chi\gi{\vec E}$, where $\gi{\vec E} = \vec E + \vec v\times\vec B$, $\vec v$ is the velocity of particles in the material 
and $\vec B$ is the magnetic field strength. 
It was believed that this new form is Galilean invariant when $|\vec v|$ is small (we will give a more precise meaning of ``slow'' later).
Since the new $\chi$ is the same as the old $\chi$ which can be determined by the same static or quasi-static experimental characterization,
this treatment introduce no extra difficulty to the material constants determination. 
All it requires is a new set of governing equations that is written in terms of $\gi{\vec E}$ in stead of $\vec E$, 
which is one of the main task of studying the Galilean invariant formulation of electromagnetism, 
and which is also what we try to establish in this note. 
$\gi{\vec E}$ will be shown to be the Galilean invariant form of $\vec E$.

In the above example, $\vec P$ is presumed to be Galilean invariant, 
which is actually true only within the non-relativistic limit, but not within the semi-relativistic limit.
The terminology of non- and semi-relativity is borrowed from \citet{Hutter}. 
According to \citet{Hutter}, non-relativistic limit means that ``in MKSA-units terms containing a $c^{-2}$-factor are neglected''.
Here, $c$ is the speed of light in vacuum.
If ``terms of order $V^2/c^2$ are neglected ($V=$ velocity of particle in the body), while those containing $c^{-2}$-factor are kept, 
we call such approximations semi-relativistic''.
Even though various existing theories can be nicely unified within in the non-relativistic limit, as proven by \citet[Chapter 3]{Hutter},
the non-relativistic limit does not have a convincing physical meaning, because physical laws should not
depend on the choice of unit system. The exact form of a term may be different in different unit systems, 
but the significance of such term should be the same in all unit systems.
The semi-relativistic limit is therefore a more proper approximation of classical electromagnetism when the velocity field is small but not negligible.
In this note, we focus on the semi-relativistic limit.
The purpose of this note is two-fold:
\begin{enumerate}[(i)]
\item It is easier to form a Galilean invariant constitutive model by Galilean invariant field variables.
\item It is even better to have governing equations also in terms of these Galilean invariant field variables.
\end{enumerate}

\section{Terminologies}
We summarize our terminology in this section. In this note, we use the indices notation for four dimensional tensors (including rank-1 tensors, {\it i.e.} vectors) and the direct notation
for three dimensional ones. By default, a three dimensional vector is represented by a 3-by-1 column vector.
The frequently used field variables are listed in Table~\ref{tbl:variables}.
\begin{table}[h]
\centering
\small
\caption{Frequently used field variables and constants in classical electromagnetism}\label{tbl:variables}
\begin{threeparttable}[b]
\begin{tabular}{ccl|ccl}
\hline\hline
notation & dimension & description & notation & dimension & description \\
\hline
$\phi$ & $\mathbb R$ & scalar potential & $\vec A$ & $\mathbb R^3$ & vector potential\\
$\rho$ & $\mathbb R$ & charge density & $\vec J$ & $\mathbb R^3$ & current density\\
$\rho_{\rm f}$ & $\mathbb R$ & free charge density & $\vec J_{\rm f}$ & $\mathbb R^3$ & free current density\\
$\vec E$ & $\mathbb R^3$ & electric field strength & $\vec B$ & $\mathbb R^3$ & magnetic field strength\\
$\vec D$ & $\mathbb R^3$ & electric displacement & $\vec H$ & $\mathbb R^3$ & magnetizing force\\
$\vec P$ & $\mathbb R^3$ & electric polarization & $\vec M$ & $\mathbb R^3$ & magnetization\\
$\epsilon_0$ & constant & vacuum permittivity\tnote{[1]} & $\mu_0$ & constant & vacuum permeability\tnote{[1]}\\
\hline\hline
\end{tabular}
\begin{tablenotes}
\item [{[1]}] $\epsilon_0$ and $\mu_0$ satisfy $\epsilon_0\mu_0c^2 = 1$.
\end{tablenotes}
\end{threeparttable}
\end{table}

In this note, we define the gradient and the curl of an arbitrary three dimensional vector field $\vec f$, and the divergence of a three dimensional tensor (of rank 2) field $\tens F$ as
\[
\begin{array}{ll}
\text{the gradient of a vector: } & (\nabla\vec f)_{ij} = \partial_if_j, \\
\text{the curl of a vector: } & (\nabla\times\vec f)_i = \epsilon_{ijk}\partial_jf_k, \\
\text{the divergence of a tensor: } & (\nabla\cdot\tens F)_i = \partial_jF_{ji}.
\end{array}
\]
Above, $\epsilon_{ijk}$ is the third-order Levi-Civita permutation operator. For three dimensional vectors and tensors, we do not
distinguish the superscript and subscript indices.

\begin{definition}\label{gtrans}
A \emph{Galilean transformation} of the three dimensional Euclidean space, $\mathbb R^3$, is the composition of a uniform translation, $\vec c\in\mathbb R^3$, and a rigid body rotation, $\tens Q\in SO(3)$. Here, $SO(3)$ is the set of all matrix representations of proper rotations. A point whose position in the original frame is $\vec x$ has the coordinate in the transformed frame $\vec x'$, which is
\begin{equation}
\vec x' = \tens Q \vec x + \vec c.
\end{equation} 
We denote such a transformation $\mathscr G(\tens Q, \vec c)$.
\end{definition}

\begin{definition}\label{def:ginv}
A scalar field $g$, a vector field $\vec g$ or a (second-rank) tensor field $\tens G$ in the three dimensional Euclidean space, $\mathbb R^3$, is called a Galilean invariant (contravariant) scalar, vector or tensor field, or simply a \emph{Galilean invariant}, if it transforms under a Galilean transformation, $\mathscr G(\tens Q,\vec c)$, according to the rules
\begin{subequations}\begin{align}
\label{ginv-1}
g'(\vec x') &= g(\vec x),\\ 
\label{ginv-2}
\vec g'(\vec x') &= \tens Q\vec g(\vec x),\\ 
\label{ginv-3}
\tens G'(\vec x') &= \tens Q\tens G(\vec x)\tens Q^T,
\end{align}\end{subequations}
where $\tens Q^T$ is the transpose of $\tens Q$. A formula (equation) is called Galilean invariant, if its form remains no change under any Galilean transformation. 
A Galilean invariant theory is a theory in which all formulas are Galilean invariant.
\end{definition}

\begin{definition}\label{spacetime}
\emph{Minkowski spacetime} is a four dimensional vector space equipped with the \emph{Minkowski metric}
\begin{equation}\label{mmetric}
\eta_{\mu\nu} = 
\begin{pmatrix}
1 & 0 \\  0 & - \mathbb{I}
\end{pmatrix},
\end{equation}
where $\mathbb I\in\mathbb R^{3\times3}$ is the 3-by-3 identity matrix. A point in the spacetime, also called an \emph{event}, is represented by
\begin{equation}
\mathcal X^\mu = \begin{pmatrix}
ct \\ \vec x
\end{pmatrix} \equiv \{ct; \vec x\},
\end{equation}
where $c$ is the speed of light in vacuum, $t\in\mathbb R$ and $\vec x = \begin{pmatrix}
x_1 \\ x_2 \\ x_3
\end{pmatrix} \in \mathbb R^3$. 
\end{definition}

\begin{remark}
$\mathcal X^\mu$ denotes also the \emph{four-displacement} from the origin to the event.
\end{remark}

\begin{definition}\label{def:ltrans}
A \emph{Lorentz transformation} of the Minkowski spacetime is represented by a 4-by-4 matrix ${\Lambda^\mu}_\nu$ satisfying
\begin{equation}\label{ltransreq}
\eta_{\mu\nu}{\Lambda^\mu}_\alpha {\Lambda^\mu}_\beta = \eta_{\alpha\beta},
\end{equation}
such that an event whose four-displacement in the original frame is $\mathcal X^\mu$ has the coordinate in the transformed frame, $\mathcal X'^\mu$, as following:
\begin{equation}
(\mathcal X')^\mu = {\Lambda^\mu}_\nu \mathcal X^\nu,
\end{equation} 
We denote such a transformation $\mathscr L({\Lambda^\mu}_\nu)$.
\end{definition}

\begin{definition}\label{def:linv}
A scalar field $\psi$, a four-vector field $g^\mu$ or a four-tensor field $\mathcal G^{\mu\nu}$ in the Minkowski spacetime is called a Lorentz invariant (contravariant) scalar, four-vector or four-tensor field, or simply a \emph{Lorentz invariant}, if it transforms under a Lorentz transformation, $\mathscr L({\Lambda^\mu}_\nu)$, according to the rules
\begin{subequations}\begin{align}
\label{linv-1}
\psi'(\mathcal X') &= \psi(\mathcal X), \\
\label{linv-2}
{g'}^{\mu}(\mathcal X') & = {\Lambda^\mu}_\nu g^\nu(\mathcal X), \\
\label{linv-3}
{\mathcal G'}^{\mu\nu}(\mathcal X') &= {\Lambda^\mu}_\alpha{\Lambda^\nu}_\beta\mathcal G^{\alpha\beta}(\mathcal X).
\end{align}\end{subequations}
A formula (equation) is called Lorentz invariant (or covariant), if its form remains no change under any Lorentz transformation. A Lorentz invariant theory is a theory in which all formulas are Lorentz invariant.
\end{definition}

\begin{remark}
If a vector field $g^\mu$ and a tensor field $\mathcal G^{\mu\nu}$ are Lorentz invariant, their covariant counterpart $g_\mu$ and $\mathcal G_{\mu\nu}$ satisfy
\begin{subequations}\begin{align}
g_\mu &= {\Lambda_\mu}^\nu g_\nu,\\
{\mathcal G'}_{\mu\nu}(\mathcal X') &= {\Lambda_\mu}^\alpha{\Lambda_\nu}^\beta\mathcal G_{\alpha\beta}(\mathcal X),
\end{align}\end{subequations}
where ${\Lambda_\mu}^\nu = \left({\Lambda^\nu}_\mu\right)^T$.
\end{remark}

Let us finishing this section by the definition of the semi-relativistic limit, a core concept of the present note.
\begin{definition}\label{def:semi-rel}
A physical problem in Minkowski spacetime is considered to be within the \emph{semi-relativistic limit}, if the following three conditions are satisfied:
\newcounter{saveenum}
\begin{enumerate}[(I)]
\item All measurements are done in the set of inertial frames moving relative to each other with velocities, $\vec u$, such that $|\vec u|/c\lesssim\epsilon$ for any pair of frames.
\label{I}
\item In any inertial frame, the velocity field of particles (field sources) satisfies $|\vec v|/c\lesssim\epsilon$.
\label{II}
\setcounter{saveenum}{\value{enumi}}
\end{enumerate}
We denote the typical value of a physical variable by [\textbullet]. Let $[x]$ and $[t]$ be the length and time scales of the problem so that for any scalar, vector or tensor field $f$ under consideration, 
\begin{equation}\label{sr-cons}
[\partial_t^m\nabla^n f]\sim[f]/[t]^m[x]^n,
\end{equation}
for every $0\leqslant m,n<+\infty$, where $\partial_t$ and $\nabla$ are respectively the time and spatial derivative operators.
Clearly, in the inertial frame under focus, $[v]\sim[x]/[t]$, and
\begin{enumerate}[(I)]
\setcounter{enumi}{\value{saveenum}}
\item the length and time scales of the problem under interest satisfy $[x]\lesssim\epsilon c[t]$.
\label{III}
\setcounter{saveenum}{\value{enumi}}
\end{enumerate}
Here, $\epsilon<1$ is such a small number that the tolerance of the theory is larger than $\mathcal O(\epsilon^2)$, $i.e.$ terms of order $\mathcal O(\epsilon^2)$ can be dropped.
\end{definition}
\begin{remark}
It may look like that the condition~\eqref{II} is just a consequence of \eqref{III}, but actually it is not. The conditions \eqref{I} and \eqref{II} are the determination of the small parameter $\epsilon$, while \eqref{III} is a constrain on other field variables and their derivatives to be suitable for a semi-relativistic setup.
\end{remark}
\begin{remark}\label{rmk:1}
Let $[\partial_t f]$ and $[\nabla f]$ represent respectively the typical value of their components in the time and spatial derivatives. Thus, if $\vec f$ is a vector field, we think $[\nabla f]\equiv[\nabla\vec f]\sim[\nabla\times\vec f]\sim[\nabla\cdot\vec f]$. If $\tens F$ is a tensor field, $[\nabla F]\equiv[\nabla\tens F]\sim[\nabla\cdot\tens F]$.
\end{remark}
\begin{remark}
A problem in which the constrain \eqref{sr-cons} cannot be satisfied for all field variables and for a single pair of $[x]$ and $[t]$ is considered to be \emph{multiscale}.
The Galilean invariant formulation of electromagnetism in multiscale problems will not be covered by the present note.
\end{remark}
\begin{corollary}
Two useful corollaries derived from Definition~\ref{def:semi-rel} are
\begin{enumerate}[(I)]
\setcounter{enumi}{\value{saveenum}}
\item for any scalar, vector or tensor field $f$ within the semi-relativistic limit, the scales of its time and spatial derivatives satisfy $[\partial_t f]  \lesssim \epsilon c [\nabla f]$,
 at almost every instance and almost everywhere.
\label{IV}
\item for any scalar, vector or tensor field $f$ within the semi-relativistic limit, $[\partial_t f]  \sim [\nabla f][v] \sim [f][\nabla v]$,
 at almost every instance and almost everywhere.
\label{V}
\end{enumerate}
\end{corollary}
\begin{proof}
Using the fact that $[\partial_tf] \sim [f]/[t]$, $[\nabla f]\sim[f]/[x]$, $[v]\sim[x]/[t]$ and conditions given in Definition~\ref{def:semi-rel}.
\end{proof}

\section{Galilean invariants of a Lorentz tensor}

The difference between non- and semi-relativistic Galilean invariant form of electromagnetic variables is best illustrated
by the Galilean invariant form of $\vec B$ field.
In the non-relativistic limit $\gi{\vec B} = \vec B$ \citep{Pen1967, Hutter, Kovetz, Ste2009},
while in the semi-relativistic limit $\gi{\vec B} = \vec B - \vec v\times\vec E/c^2$ \citep{Hutter, Tiersten}.
Obviously, two choices are equivalent within the non-relativistic limit, but 
such limit has litter physical meaning, as discussed before.
Actually, $\gi{\vec B}=\vec B$ can be recovered in another limiting case even in the semi-relativistic setup, namely the \emph{magnetic limit},
which will be introduced in Sect.~\ref{sec:limits}.
The discrepancy between the two choices of $\gi{\vec B}$ originates in the different way or applying a Galilean transformation to a Lorentz (contravariant or covariant) tensor.

In fact, a rigid body rotation gives no barrier between Galilean and Lorentz invariance properties.
It is the boost $\vec u$ that gives us the trouble.
So, in the following discussion we will only consider the pure Galilean boost, $\mathscr G(\mathbb I, -\vec ut)$.
Our derivation of the Galilean invariant form of $\vec B$ starts with the Lorentz boost
\begin{equation}\label{lboost}
{{\Lambda_{\rm b}}^\mu}_\nu = 
\begin{pmatrix}
\gamma_u & -\gamma_u\vec u^T/c \\
-\gamma_u\vec u/c & \quad\mathbb I + (\gamma_u-1)\vec u\vec u^T/u^2
\end{pmatrix},
\end{equation}
where $\gamma_u = 1/\sqrt{1-u^2/c^2}$ is the \emph{Lorentz factor} of speed $u$.
Here, $\vec u\vec u^T$ can also be written as $\vec u\otimes\vec u$.
Clearly, in the semi-relativistic limit, $\gamma_u = 1$. 
The semi-relativistic Lorentz boost is now
\begin{equation}\label{lboost-semirel}
{{\Lambda_{\rm b-sr}}^\mu}_\nu = 
\begin{pmatrix}
1 & -\vec u^T/c \\
-\vec u/c & \quad\mathbb I
\end{pmatrix}
\end{equation}
Such transformation was called the ``extended Lorentz transformation'' by \citet{Hutter}.
Applying this transformation to the \emph{electromagnetic four-field tensor} gives
\begin{equation}\label{ofmunu}
 \mathscr F'^{\mu\nu} = {{\Lambda_{\rm nr}}^\mu}_{\alpha} {{\Lambda_{\rm nr}}^\nu}_{\beta} \mathscr F^{\alpha\beta}
\end{equation} 
Then recall the relationship between $\mathscr F$ and $\vec E$ and $\vec B$ \citep{Minkowski, Jackson},
\begin{equation}\label{four-field}
\mathscr F^{\mu\nu} = 
\begin{pmatrix}
0 & -\vec E^T/c \\
\vec E/c & \tens W(\vec B)
\end{pmatrix}.
\end{equation}
$\tens W(\vec a) : \mathbb R^3 \to \mathbb R^{3\times3}$ is the antisymmetric matrix with the skew axis $\vec a\in\mathbb R^3$.
It has the property that $\tens W(\vec a)\vec b = \vec a\times\vec b$ for any $\vec b\in\mathbb R^3$.
Eqn.~\eqref{ofmunu} gives $\vec E'  = \tens Q(\vec E + \vec u\times\vec B$) and $\vec B' = \tens Q\vec B$. 
Following the definition \eqref{ginv-2} and 
the fact that the velocity field transforms under $\mathscr G(\tens Q, -\tens Q\vec ut)$ according to $\vec v' = \tens Q(\vec v - \vec u)$, 
the Galilean invariant forms are obtained as $\gi{\vec E} = \vec E + \vec v\times\vec B$ and $\gi{\vec B} = \vec B$.
Applying the transformation \eqref{lboost-semirel} to the \emph{electromagnetic four-field tensor} gives
\begin{equation}\label{ofmunu}
 \mathscr F'^{\mu\nu} = {{\Lambda_{\rm b-sr}}^\mu}_{\alpha} {{\Lambda_{\rm b-sr}}^\nu}_{\beta} \mathscr F^{\alpha\beta}
\end{equation} 
Then recall the relationship between $\mathscr F$ and $\vec E$ and $\vec B$ \citep{Minkowski, Jackson},
\begin{equation}\label{four-field}
\mathscr F^{\mu\nu} = 
\begin{pmatrix}
0 & -\vec E^T/c \\
\vec E/c & \tens W(\vec B)
\end{pmatrix}.
\end{equation}
$\tens W(\vec a) : \mathbb R^3 \to \mathbb R^{3\times3}$ is the antisymmetric matrix with the skew axis $\vec a\in\mathbb R^3$.
It has the property that $\tens W(\vec a)\vec b = \vec a\times\vec b$ for any $\vec b\in\mathbb R^3$.
Eqn.~\eqref{ofmunu} gives $\vec E'  = \tens Q(\vec E + \vec u\times\vec B$) and $\vec B' = \tens Q\vec B$. 
Following the definition \eqref{ginv-2} and 
the fact that the velocity field transforms under $\mathscr G(\tens Q, -\tens Q\vec ut)$ according to $\vec v' = \tens Q(\vec v - \vec u)$, 
the Galilean invariant forms are obtained as $\gi{\vec E} = \vec E + \vec v\times\vec B$ and $\gi{\vec B} = \vec B - \vec v\times\vec E/c^2$.

\begin{remark}
An equivalent way of seeing this is starting with the transformed $\vec E$ and $\vec B$ by a fully relativistic
Lorentz boost, Eqn.~\eqref{lboost}. The results are \citep[Eqn.~(11.149)]{Jackson}
\[
\begin{array}{l}
\vec E' = \gamma_u\left(\vec E + \vec u\times\vec B\right) - \left.\gamma_u^2\vec u(\vec u\cdot\vec E)\right/c^2(\gamma_u+1) \\
\vec B' = \gamma_u\left(\vec B - \vec u\times\vec E/c^2\right) - \left.\gamma_u^2\vec u(\vec u\cdot\vec B)\right/c^2(\gamma_u+1).
\end{array}
\]
Then applying the semi-relativistic approximation, in particular $[\vec u(\vec u\cdot\vec E)/c^2] \lesssim \epsilon^2[\vec E]$, yields the same expressions as above.
\end{remark}

\begin{remark}
Applying the transformation \eqref{lboost-semirel} to an event $\mathscr X^\mu = \{ct; \vec x\}$ in the Minkowski spacetime, we get
\[
\begin{pmatrix}
ct' \\ \vec x'
\end{pmatrix} = X'^\mu = {{\Lambda_{\rm b-sr}}^\mu}_\nu \mathscr X^\nu = 
\begin{pmatrix}
1 & -\vec u^T/c \\
-\vec u/c & \quad\mathbb I
\end{pmatrix}
\begin{pmatrix}
ct \\ \vec x
\end{pmatrix}
=\begin{pmatrix}
ct - \vec u^T\vec x/c \\
\vec x - \vec u t
\end{pmatrix}
\]
Here, $\vec u^T\vec x = \vec u\cdot\vec x$. 
Using semi-relativistic limit conditions \eqref{I} and \eqref{III},
we have $[\vec u\cdot\vec x/c] \lesssim \epsilon^2 [c t]$.
Hence, we have proven that the transformation \eqref{lboost-semirel} is equivalent to
the Galilean boost $\mathscr G(\mathbb I, -\vec ut)$ within the semi-relativistic limit.
\end{remark}

Now, we introduce the procedure of obtaining the Galilean invariant forms of a general Lorentz tensor (of rank 1 or 2).
\begin{proposition}\label{prop}
Let the four dimensional vector $g^\mu$ and the four dimensional matrix $\mathscr G^{\mu\nu}$ such that
\[
g^\mu = \begin{pmatrix}g^0\\\vec g^1\end{pmatrix},\qquad
\mathscr G^{\mu\nu} = 
\begin{pmatrix}
G^0 & (\vec G^1)^T \\ 
\vec G^2 & \tens G^3
\end{pmatrix},
\]
for some $g^0,G^0\in\mathbb R$, $\vec g^1, \vec G^1$ and $\vec G^2\in\mathbb R^3$, $\tens G^3\in\mathbb R^{3\times 3}$,
be Lorentz invariant (contravariant) tensors. Then
\begin{equation}\label{gformvec}
\begin{cases}
\gi{g^0} \equiv g^0 - \vec v\cdot\vec g^1/c^2 \\
\gi{\vec g^1} \equiv \vec g^1 - g^0\vec v
\end{cases}
\end{equation}
\begin{equation}\label{gformtens}
\begin{cases}
\gi{G^0} \equiv G^0 - \vec G^1\cdot\vec v/c - \vec v\cdot\vec G^2/c + \vec v\cdot\tens G^3\vec v/c^2 \\
\gi{\vec G^1} \equiv \vec G^1 - G^0\vec v/c + (\vec v\cdot\vec G^2)\vec v/c^2 - \tens G^3\vec v/c \\
\gi{\vec G^2} \equiv \vec G^2 - G^0\vec v/c + (\vec v\cdot\vec G^1)\vec v/c^2 - \tens G^3\vec v/c \\
\gi{\tens G^3} \equiv \tens G^3 + G^0\vec v\otimes\vec v/c^2 - \vec v\otimes\vec G^1/c -\vec G^2\otimes\vec v/c 
\end{cases}
\end{equation}
are Galilean invariant within the semi-relativistic limit. 
\end{proposition}
\begin{proof}
Because $g^\mu$ and $\mathscr G^{\mu\nu}$ are Lorentz invariant, the semi-relativistic Lorentz boost can be
directly applied to them and yields
\[
g'^\mu = {{\Lambda_{\rm b-sr}}^\mu}_\nu g^\nu = 
\begin{pmatrix}
1 & -\vec u^T/c \\
-\vec u/c & \quad\mathbb I
\end{pmatrix}
\begin{pmatrix}g^0\\\vec g^1\end{pmatrix}
=
\begin{pmatrix}
g^0 - \vec u^T\vec g^1/c^2 \\
\vec g^1 - g^0\vec u
\end{pmatrix},
\]
\[
\begin{split}
\mathscr G'^{\mu\nu} &= {{\Lambda_{\rm b-sr}}^\mu}_\alpha{{\Lambda_{\rm b-sr}}^\mu}_\beta\mathscr G^{\alpha\beta}
= \begin{pmatrix}
1 & -\vec u^T/c \\
-\vec u/c & \quad\mathbb I
\end{pmatrix}
\begin{pmatrix}
G^0 & (\vec G^1)^T \\ 
\vec G^2 & \tens G^3
\end{pmatrix}
\begin{pmatrix}
1 & -\vec u^T/c \\
-\vec u/c & \quad\mathbb I
\end{pmatrix}\\
& = \begin{pmatrix}
G^0 - (\vec G^1)^T\vec u/c - \vec u^T\vec G^2/c + \vec u^T\tens G^3\vec u/c^2 &
\quad (\vec G^1)^T - G^0\vec u^T/c + \vec u^T\vec G^2\vec u^T/c^2 - \vec u^T\tens G^3/c\\
\vec G^2 - G^0\vec u/c + \vec u(\vec G^1)^T\vec u/c^2 - \tens G^3\vec u/c & 
\quad \tens G^3 + G^0\vec u\vec u^T/c^2 - \vec u(\vec G^1)^T/c -\vec G^2\vec u^T/c 
\end{pmatrix}.
\end{split}
\]
Then using the fact that under the same Galilean boost, $\vec v \mapsto \vec v' = \vec v - \vec u$, we are done. 
\end{proof}

We name these asterisked variables the \emph{Galilean invariant forms}, or simply the \emph{Galilean invariants}, of the corresponding components of the Lorentz vector $g^\mu$ and the Lorentz tensor $\mathscr G^{\mu\nu}$. This terminology is borrowed from \citet{Kovetz}, although our forms are slightly different to his.

The Galilean invariants corresponding to the components of the \emph{covariant} counterpart of $g^\mu$ and $\mathscr G^{\mu\nu}$, 
denoted 
\[
g_\mu = \begin{pmatrix} g_0&~\vec g_1\end{pmatrix}\quad\text{and}\quad
\mathscr G_{\mu\nu} = \begin{pmatrix}
G_0 & \vec G_1^T \\ 
\vec G_2 & \tens G_3
\end{pmatrix}
\]
respectively, can be obtained by using the relations $g_0 = g^0$,  $\vec g_1 = -\vec g^1$, $G_0 = G^0$, $\vec G_{1,2} = - \vec G^{1,2}$, and $\tens G_3 = \tens G^3$
as a consequence of the Minkowski metric, Eqn.~\eqref{mmetric}. An important application of this covariant version of Proposition~\ref{prop} is 
to get the Galilean invariant forms of the covariant four-gradient operator
$\partial_\mu = (\partial_t/c, \nabla^T)$, where $\nabla^T = (\partial_1, \partial_2, \partial_3)$ is the gradient operator in $\mathbb R^3$.
Using Eqn.~\eqref{gformvec} and the aforementioned relations, in addition to the semi-relativistic limit condition \eqref{IV}, 
we have
\begin{equation}\label{gform-grad}
\begin{array}{l}
\gi{\partial_t} = \partial_t + \vec v\cdot\nabla, \\
\gi{\nabla} = \nabla.
\end{array}
\end{equation}
Now, we can restate the condition~\eqref{IV} as ``in the semi-relativistic limit, the Lorentz covariant four-gradient operator is \emph{ultra space-like}'', {i.e.} the spatial component is much larger than the time component.

Another derivative operator which has been found to be useful in the Galilean invariant formulation of non-relativistic (and semi-relativistic) electromagnetism is the \emph{flux derivative} \citep{Hutter, Kovetz, Ste2009}, which is defined as
\begin{equation}\label{flux-d}
\grave{\vec a} = \left(\partial_t + \vec v\cdot\nabla\right)\vec a + (\nabla\cdot\vec v)\vec a - \vec a\cdot\nabla\vec v
\end{equation}
for a vector field $\vec a\in\mathbb R^3$ with sufficient regularity. Since all components of the gradient of $\vec v$ are Galilean invariant,
the flux derivative is a Galilean invariant operator. It has the property that for an arbitrary sufficiently regular surface $\mathscr S\in\mathbb R^3$,
\[
\int_{\mathscr S}\grave{\vec a}\cdot\vec n ds = \dfrac{d}{dt}\int_{\mathscr S}\vec a\cdot\vec nds.
\]
Here, $\vec n$ is the unit normal of $\mathscr S$.

Clearly, a \emph{volume derivative}, defined as
\begin{equation}\label{vol-d}
\check{\vec a} = \left(\partial_t + \vec v\cdot\nabla + \nabla\cdot\vec v\right)\vec a
\end{equation}
for a vector field $\vec a\in\mathbb R^3$ with sufficient regularity, is also Galilean invariant. 
It has the property that for an arbitrary open bounded domain $\mathscr B\in\mathbb R^3$,
\[
\int_{\mathscr B}\check{\vec a}dv = \dfrac{d}{dt}\int_{\mathscr B} \vec adv.
\]
 
\begin{remark}
We cannot cancel terms like $\vec v\cdot\tens G_3\vec v/c^2$ in \eqref{gformvec} and \eqref{gformtens} right away because the relative sizes between the components such as $G_0$ and $\tens G_3$ are unknown.
\end{remark}
\begin{remark}
From the condition \eqref{V}, we can see that for a field $f$, $[\grave f]\sim[\check f]\sim[\partial_t f]$.
\end{remark}

\section{Galilean invariants in electromagnetism}

In this section, we are going to show that the classical electromagnetism can be reformulated to be Galilean invariant within the semi-relativistic limit, using the concept of Galilean invariants and the procedure described in Proposition~\ref{prop}.
Through the resulting formulation we will see that the Galilean invariants introduced in Sect.~\ref{sec:gi} are more proper choices of 
independent variables than the original field variables to be used in the theory of electromagnetic continua.

\subsection{Galilean invariant variables}\label{sec:gi}
In (special relativistic) classical electromagnetism \citep{Lorentz, Landau_Field, Landau_Continuous, Jackson}, all fields and sources can be represented by 
Lorentz invariant four-vectors and four-tensors whose Galilean invariants can be obtained by applying Proposition~\ref{prop}.
We summarize them in Table~\ref{tbl:li-vec-tens}.
\begin{table}[h]
\centering
\scriptsize
\caption{Lorentz invariant (contravarient) four-vectors and four-tensors in electromagnetism and their Galilean invariants}\label{tbl:li-vec-tens}
\begin{threeparttable}[b]
\begin{tabular}{rclcl}
\hline\hline
&4-D variables & descriptions & block matrix representations &  Galilean invariants\\
\hline \vspace{-0.5em}\\
1)&$\mathscr A^\mu$ & four-potential & $\{\phi/c;~\vec A\}$ & 
$\begin{array}{l} \gi{\phi} \equiv \phi - \vec v\cdot\vec A \\ \gi{\vec A} \equiv \vec A - \phi\vec v/c^2\end{array}$
\vspace{0.5em}\\
2)&$\mathscr J^\mu$ & four-current & $\{c\rho;~\vec J\}$ & 
$\begin{array}{l} \gi{\rho} \equiv \rho - \vec v\cdot\vec J/c^2 \\ \gi{\vec J} \equiv \vec J - \rho\vec v\end{array}$
\vspace{0.5em}\\
3)&$\mathscr J_{\rm f}^\mu$ & free four-current & $\{c\rho_{\rm f};~\vec J_{\rm f}\}$ & 
$\begin{array}{l} 
\gi{\rho_{\rm f}} \equiv \phi_{\rm f} - \vec v\cdot\vec J_{\rm f}/c^2 \\ 
\gi{\vec J_{\rm f}} \equiv \vec J_{\rm f} - \rho_{\rm f}\vec v
\end{array}$\vspace{0.5em}\\
4)&$\mathscr F^{\mu\nu}$ & field strength tensor & $\begin{pmatrix} 0 & -\vec E^T/c \\ \vec E/c & \tens W(\vec B) \end{pmatrix}$ & 
$\begin{array}{l} \gi{\vec E} \equiv \vec E + \vec v\times\vec B \\ \gi{\vec B} \equiv \vec B - \vec v\times\vec E/c^2\end{array}$ 
\vspace{0.5em}\\
5)&$\mathscr D^{\mu\nu}$ & free field strength& $\begin{pmatrix} 0 & -c\vec D^T \\ c\vec D & \tens W(\vec H) \end{pmatrix}$ &
$\begin{array}{l} \gi{\vec D} \equiv \vec D + \vec v\times\vec H/c^2 \\ \gi{\vec H} \equiv \vec H - \vec v\times\vec D\end{array}$
\vspace{0.5em}\\
6)&$\mathscr M^{\mu\nu}$ & bounded field strength & $\begin{pmatrix} 0 & c\vec P^T \\ -c\vec P & \tens W(\vec M) \end{pmatrix}$ &
$\begin{array}{l} \gi{\vec P} \equiv \vec P - \vec v\times\vec M/c^2 \\ \gi{\vec M} \equiv \vec M + \vec v\times\vec P\end{array}$
\vspace{0.5em}\\\hline \vspace{-0.5em}\\
7)&$f^\mu$ & four-force & $\{h_{\rm J}/c;~\vec f_{\rm L}\}$\tnote{[1]} & 
$\begin{array}{l} 
\gi{h_{\rm J}} \equiv h_{\rm J} - \vec v\cdot\vec f_{\rm L} \\ 
\gi{\vec f_{\rm L}} \equiv \vec f_{\rm L} - h_{\rm J}\vec v/c^2
\end{array}$
\vspace{0.5em}\\
8)&$\mathscr T^{\mu\nu}$ & energy momentum tensor & $\begin{pmatrix} w & \vec S^T/c \\ \vec S/c & -\tens T \end{pmatrix}$\tnote{[2]} &
$\begin{array}{l} 
\gi{w} \equiv w - 2\vec S\cdot\vec v/c^2 \\ 
\gi{\vec S} \equiv \vec S - w\vec v + \tens T\vec v\\
\gi{\tens T} \equiv \tens T + (\vec v\otimes\vec S + \vec S\otimes\vec v)/c^2
\end{array}$
\vspace{0.5em}\\
\hline\hline
\end{tabular}
\begin{tablenotes}
\item [{[1]}] $h_{\rm J} = \vec J\cdot\vec E$ is the \emph{Joule heating}, and $\vec f_{\rm L} = \rho\vec E + \vec J\times\vec B$ is the \emph{Lorentz force}.
\item [{[2]}] $w = \epsilon_0|\vec E|^2/2 + |\vec B|^2/2\mu_0$ is the \emph{field energy density}, $\vec S = \vec E\times\vec B/\mu_0$ is the \emph{Poynting vector},
and $\tens T = \epsilon_0\vec E\otimes \vec E + \vec B\otimes\vec B/\mu_0 - w \mathbb I$ is the \emph{Maxwell stress}.
\end{tablenotes}
\end{threeparttable}
\end{table}

\begin{remark}
{The Galilean invariants in all but the last two rows of Table~\ref{tbl:li-vec-tens} are similar to those in the literature \citep{Hutter, Kovetz, Ericksen2007, Ste2009} except the terms containing $c^{-2}$-factor. This is because Table~\ref{tbl:li-vec-tens} is based on the semi- rather than non-relativisitic limit}.
\end{remark}
\begin{remark}
The Galilean invariants in the last two rows of Table~\ref{tbl:li-vec-tens} are can be recovered by replacing all variables in their definitions (see the table notes) with the corresponding Galilean invariants given in the first six rows, and applying the semi-relativistic approximation. 
{\it e.g.} $\gi{h_J} = \gi{\vec J}\cdot\gi{\vec E}$, and $\gi{\tens T} = \epsilon_0\gi{\vec E}\otimes \gi{\vec E} +\gi{\vec B}\otimes\gi{\vec B}/\mu_0 - \gi{w} \mathbb I$. 
\end{remark}

\begin{remark}
Although electromagnetism is based on the fields in the top six rows of Table~\ref{tbl:li-vec-tens}, what really plays important roles in 
the thermomechanical theory are the last two rows.
\end{remark}

\subsection{Galilean invariant formulation}

The classical electromagnetism can be formulated completely be Lorentz invariant fields and operators \citep{Lorentz, Landau_Field, Landau_Continuous, Jackson}. 
The formulation has a 3D representation, including Maxwell's equations.
But it is non-objective under Galilean transformations of $\mathbb R^3$.
By algebraic manipulations, these 3D equations can be 
rewritten in terms of only the Galilean invariants obtained in Table~\ref{tbl:li-vec-tens}
and the Galilean invariant derivative operators \eqref{gform-grad} and \eqref{flux-d}.
The result is the Galilean invariant formulation of classical electromagnetism within the semi-relativistic limit,
which is summarized in Table~\ref{tbl:lc-form}.

\begin{table}[h]
\centering
\scriptsize
\caption{\small 4D Lorentz covariant electromagnetism and its 3D Galilean invariant formulation}\label{tbl:lc-form}
\begin{threeparttable}[b]
\begin{tabular}{rlllc}
\hline\hline
&\multicolumn{1}{c}{Lorentz invariant 4D} & \multicolumn{1}{c}{non-objective 3D} &  \multicolumn{2}{c}{Galilean invariant 3D formulation}\\
&\multicolumn{1}{c}{formulation}&\multicolumn{1}{c}{formulation}&local form&{global form}\\
\hline \vspace{-0.5em}\\
1)  & $\mathscr F^{\mu\nu} = \partial^\mu\mathscr A^\nu - \partial^\nu\mathscr A^\mu$ &
$\begin{array}{l}
\vec E = -\nabla\phi - \partial_t\vec A \\
\vec B =  \nabla\times\vec A
\end{array}$ &
$\begin{array}{l}
\gi{\vec E} = -\nabla\gi{\phi} - (\gi{\partial_t} + \nabla\vec v)\gi{\vec A} \\
\gi{\vec B} = \nabla\times\gi{\vec A} + \gi{\phi}\nabla\times\vec v/c^2
\end{array}$ &
\vspace{0.5em}\\
2) & $\epsilon^{\alpha\beta\mu\nu}\partial_\beta\mathscr F_{\mu\nu} = 0$\tnote{[1]} &
$\begin{array}{l}
\nabla\cdot\vec B = 0 \\
\nabla\times\vec E = - \partial_t\vec B
\end{array}$ &
$\begin{array}{l}
\nabla\cdot\gi{\vec B} = 0 \\
\nabla\times\gi{\vec E} = - \grave{\gi{\vec B}}
\end{array}$ &
\eqref{global-maxwell-1}
\vspace{0.5em}\\
3) & $\partial_\mu\mathscr F^{\mu\nu} = \mu_0\mathscr J^\mu$ &
$\begin{array}{l}
\nabla\cdot\vec E = \rho /\epsilon_0 \\
\nabla\times\vec B = \mu_0\vec J + \partial_t\vec E/c^2
\end{array}$ &
$\begin{array}{l}
\nabla\cdot\gi{\vec E} = \gi{\rho}/\epsilon_0 \\
\nabla\times\gi{\vec B} = \mu_0\gi{\vec J} + \grave{\gi{\vec E}}/c^2
\end{array}$ &
\eqref{global-maxwell-2}
\vspace{0.5em}\\
4) & $\partial_\mu\mathscr D^{\mu\nu} = \mu_0\mathscr J_{\rm f}^\mu$ &
$\begin{array}{l}
\nabla\cdot\vec D = \rho_{\rm f} \\
\nabla\times\vec H = \mu_0\vec J_{\rm f} + \partial_t\vec H
\end{array}$ &
$\begin{array}{l}
\nabla\cdot\gi{\vec D} = \gi{\rho_{\rm f}} \\
\nabla\times\gi{\vec H} = \mu_0\gi{\vec J_{\rm f}} + \gi{\vec H}
\end{array}$ &
\eqref{global-maxwell-3}
\vspace{0.5em}\\
5) & $\mathscr F^{\mu\nu}/\mu_0 = \mathscr D^{\mu\nu} + \mathscr M^{\mu\nu}$ &
$\begin{array}{l}
\epsilon_0\vec E = \vec D -\vec P \\
\vec B/\mu_0 = \vec H + \vec M 
\end{array}$ &
$\begin{array}{l}
\epsilon_0\gi{\vec E} = \gi{\vec D} - \gi{\vec P} \\
\gi{\vec B}/\mu_0 = \gi{\vec H} + \gi{\vec M}
\end{array} $ &
\vspace{0.5em}\\\hline \vspace{-0.5em}\\
6) & $\partial_\nu\mathscr T^{\mu\nu} + \mathscr F_{\mu\nu}\mathscr J^\nu = 0$ &
$\begin{array}{l}
\partial_t w  + h_{\rm J} + \nabla\cdot\vec S = 0 \\
\partial_t\vec S/c^2 + \vec f_{\rm L} - \nabla\cdot\tens T =  0
\end{array}$& 
$\begin{array}{l}
\check{\gi{w}} + \gi{h_{\rm J}} + \nabla\cdot\gi{\vec S} = \gi{\tens T}\cdot\nabla\vec v \\
\check{\gi{\vec S}}/c^2 + \gi{\vec f_{\rm L}} - \nabla\cdot\gi{\tens T} = -\gi{\vec S}\cdot\nabla\vec v/c^2
\end{array}$ &
$\begin{array}{l}
\eqref{global-poynting}\\
\eqref{global-momentum}
\end{array}$
\vspace{0.5em}\\
\hline\hline
\end{tabular}
\begin{tablenotes}
\item [{[1]}] $\epsilon^{\alpha\beta\mu\nu}$ is the forth-order Levi-Civita permutation operator.
\end{tablenotes}
\end{threeparttable}
\end{table}

In the derivation of the Galilean invariant forms in Table~\ref{tbl:lc-form}, particularly those of Poynting's theorem and the momentum identity (the last row of Table~\ref{tbl:lc-form}), we used
the semi-relativistic limit condition~\eqref{II} and \eqref{IV} and Remark~\ref{rmk:1} to claim
\[
\left.\begin{array}{r}
{[\partial_t(f\cdot\vec v)/c^2]} \\
{[\vec v\cdot(\partial_tf)/c^2]}
\end{array}\right\}
\sim \dfrac{[v]}{c}\dfrac{[\partial_tf]}{c} \lesssim \epsilon^2 [\nabla f] \sim \epsilon^2 [\nabla\cdot f]
\]
for any semi-relativistic scalar, vector, or tensor field $f$.
Thus, because of the presence of $\nabla\cdot\vec S$, in the Galilean invariant Poynting's theorem, we have the approximation
$\tens T\cdot\nabla\vec v \approx \gi{\tens T}\cdot\nabla \vec v.$
Also, recall $[w]\sim[\tens T]$, the presence of $\nabla\cdot\tens T$ in the momentum identity enables us to neglect terms like 
$\partial_t(w\vec v)/c^2$, $\partial_t(\tens T\vec v)/c^2$ and $(\partial_tw)\vec v/c^2$ during the derivation.

For an arbitrary bounded open domain $\mathscr B\in\mathbb R^3$ whose boundary $\partial\mathscr B$ has
a sufficiently regular unit outer normal vector field $\vec n$, and for an arbitrary bounded surface $\mathscr S\in\mathbb R^3$ 
whose boundary $\partial\mathscr S$ has a sufficiently regular unit tangential vector field $\vec t$ such that $\mathscr S$ is always
on the left of $\vec t$, the following Galilean invariant identities hold:
\begin{align}
\label{global-maxwell-1}
&\int_{\partial\mathscr B}\gi{\vec B}\cdot\vec nds = 0, &\qquad
& \int_{\partial\mathscr S} \gi{\vec E}\cdot\vec tdl = -\dfrac{d}{dt}\int_{\mathscr S}\gi{\vec B}\cdot\vec nds,\\
\label{global-maxwell-2}
&\int_{\partial\mathscr B}\gi{\vec E}\cdot\vec nds = \int_{\mathscr B}\dfrac{\gi{\rho}}{\epsilon_0}dv, &\qquad
& \int_{\partial\mathscr S} \gi{\vec B}\cdot\vec tdl = \int_{\mathscr S}\mu_0\gi{\vec J}\cdot\vec nds-\dfrac{1}{c^2}\dfrac{d}{dt}\int_{\mathscr S}\gi{\vec E}\cdot\vec nds,\\
\label{global-maxwell-3}
&\int_{\partial\mathscr B}\gi{\vec D}\cdot\vec nds = \int_{\mathscr B}\gi{\rho_{\rm f}}dv,&\qquad
& \int_{\partial\mathscr S} \gi{\vec H}\cdot\vec tdl = \int_{\mathscr S}\gi{\vec J_{\rm f}}\cdot\vec nds-\dfrac{d}{dt}\int_{\mathscr S}\gi{\vec D}\cdot\vec nds,
\end{align}
\begin{align}
\label{global-poynting}
&\dfrac{d}{dt}\int_{\mathscr B} \gi{w} dv + \int_{\mathscr B}\gi{h_{\rm J}}dv 
+ \oint_{\partial\mathscr B}\gi{\vec S}\cdot\vec nds = \int_{\mathscr B}\gi{\tens T}\cdot\nabla\vec vdv, \\
\label{global-momentum}
&\dfrac{d}{dt}\int_{\mathscr B}\dfrac{\gi{\vec S}}{c^2} dv + \int_{\mathscr B}\gi{f_{\rm L}}  - \nabla\cdot\gi{\tens T} dv
=  - \int_{\mathscr B}\dfrac{\gi{\vec S}}{c^2}\cdot\nabla\vec vdv.
\end{align}

\begin{remark}
{The right hand sides of \eqref{global-poynting} and \eqref{global-momentum} are rarely mentioned in past literature}.
The former can be interpreted as a dissipation rate due to the (Galilean invariant) Maxwell stress.
The latter can be interpreted as a work due to the non-symmetry between the electromagnetic momentum $\gi{\vec S}/c^2$
and the mechanical momentum $\vec v$. Such non-symmetry may contribute to a non-symmetric Cauchy stress as discussed by \citet{Kovetz} and others.
\end{remark}

\begin{remark}
{
Although similar Galilean invariant Maxwell's equations have been introduced by many past works
\citep{Pen1967, Truesdell_Toupin, Maugin1988, Tiersten, Hutter, Kovetz, Ericksen2007, Ste2009}, 
the real goal of the present note is
to derive the Galilean invariant Poynting's theorem and the momentum identity in the last row of Table~\ref{tbl:lc-form}, as well
as their global forms \eqref{global-poynting} and \eqref{global-momentum}.
}
The former identity is closely related to the conservation of energy,
and the latter one is closely related to the balance of linear and angular momentum in the 
thermomechanical theory.
\end{remark}

\section{Stronger limits}\label{sec:limits}
As we mentioned earlier, the terms containing $c^{-2}$-factor cannot be brutally neglected, unless further information about 
the relative sizes between various field variables are known. In this section, we discuss some stronger limiting cases where
such relative sizes are partially known, and simplify the Galilean invariant formulation of electromagnetism in these stronger limits.

\subsection{Magnetic limits}
\begin{definition}
A semi-relativistic electromagnetic problem is said to be within the \emph{(weak) magnetic limit}, if 
\begin{equation}\label{lim-mag}
[E] \lesssim \epsilon c [B]
\end{equation}
The problem is said to be within the \emph{strong magnetic limit}, if
in addition to \eqref{lim-mag}, also
\begin{equation}\label{lim-mag-strong}
c[D]\lesssim\epsilon [H],\qquad
c[P]\lesssim\epsilon [M]
\end{equation}
\end{definition}

Some of the Galilean invariants listed in Table~\ref{tbl:li-vec-tens} can be reduced, when the magnetic limit is reached.
For example, $\gi{\vec B}=\vec B$ can be recovered. 
If it is the strong magnetic limit, we shall also have
$\gi{\vec H} = \vec H$ and $\gi{\vec M} = \vec M$, 
{which are different from the Galilean invariants
obtained in most past works \citep{Kovetz, Ericksen2007, Ste2009}, 
where the non-relativistic limit was used.
In past works, usually $\gi{\vec D}=\vec D$ and $\gi{\vec P} = \vec P$.}
But, from \eqref{lim-mag-strong} we can only see that 
${[v][H]/c^2}/{[D]} \gtrsim {[v]}/{c\epsilon} > {[v]}/{c}.$
This does not exclude the possibility that $[\vec v\times\vec H/c^2] > \epsilon^2 [\vec D]$, because $[v]/c > \epsilon^2$ is
allowed in the semi-relativistic limit.
Hence, $\gi{\vec D} = \vec D + \vec v\times\vec H/c^2$ cannot be further reduced with the magnetic limit.

We notice that according to the first row in Table~\ref{tbl:lc-form}, the magnetic limit implies 
\[
[\nabla\phi]\lesssim\epsilon c [\nabla A].
\]
According to the constrain \eqref{sr-cons} given in Definition~\ref{def:semi-rel}, above inequality is
equivalent to
\begin{equation}\label{lim-mag-2}
[\phi]\lesssim\epsilon c [A].
\end{equation}
This inequality can be treated as an alternative way of stating the magnetic limit.
Thus the second equation in the first row of Table~\ref{tbl:lc-form} is reduced to
\[
\gi{\vec B} = \nabla\times\gi{\vec A}.
\]
Eqn.~\eqref{global-maxwell-2}$_2$ and its local form given in Table~\ref{tbl:lc-form} are also reduced, because
\[
[\partial_t E]/c^2 \lesssim \epsilon [\partial_t B] / c \lesssim \epsilon^2 [\nabla B]
\]
within the magnetic limit. There are other equations can be reduced when the magnetic or the strong magnetic limit is reached.
All the reductions are listed in Table~\ref{tbl:reduction}.

\begin{table}[h]
\centering
\small
\caption{Reduction of Galilean invariant variables and equaitons within magnetic and electric limits}\label{tbl:reduction}
\renewcommand{\arraystretch}{1.3}
\begin{tabular}{l|c|c}
\hline\hline
& {\qquad}magnetic limit{\qquad} & {\qquad}electric limit{\qquad}\\
\hline
\multirow{4}{*}{(weak limit)}
& $\gi{\vec A} = \vec A$, $\gi{\vec B} = \vec B$, $\gi{\vec J} = \vec J$ & $\gi{\phi} = \phi$, $\gi{\vec E} = \vec E$, $\gi{\rho} = \rho$\\
& $\gi{\vec f_{\rm L}} = \vec f_{\rm L}$, $\gi{w} = w$, $\gi{\tens T} = \tens T$ & $\gi{h_{\rm J}} = h_{\rm J}$, $\gi{w} = w$, $\gi{\tens T} = \tens T$\\
& $\gi{\vec B} = \nabla\times\gi{\vec A}$ & $\gi{\vec E} = -\nabla\gi{\phi}$\\
& $\nabla\times\gi{\vec B} = \mu_0\gi{\vec J}$ & $\nabla\times\gi{\vec E} = 0$ \\
\hline
\multirow{2}{*}{strong limit}
& $\gi{\vec H} = \vec H$, $\gi{\vec M} = \vec M$, $\gi{\vec J_{\rm f}} = \vec J_{\rm f}$ & $\gi{\vec D} = \vec D$, $\gi{\vec P} = \vec P$, $\gi{\rho_{\rm f}} = \rho_{\rm f}$ \\
& $\nabla\times\gi{\vec H} = \gi{\vec J_{\rm f}}$ &  \\
\hline\hline
\end{tabular}
\renewcommand{\arraystretch}{1}
\end{table}

From \eqref{global-maxwell-2}$_1$ and the reduced form of \eqref{global-maxwell-2}$_2$, we have
$\epsilon_0[\nabla E]\sim[\rho]$ and $[\nabla B]\sim\mu_0[J]$, which yields
\begin{equation}\label{lim-mag-3}
c[\rho]\lesssim\epsilon[J]
\end{equation}
according to \eqref{lim-mag}.
Then we obtain the reduction $\gi{\vec J} = \vec J$, because $[v][\rho]\lesssim\epsilon^2[J]$.
Such reduction has been also found in other ``magnetic problems'', such as magnetohydrodynamics \citep[Sect.~61]{Kovetz}.
The inequality \eqref{lim-mag-3} can be treated as the third way of stating the magnetic limit.

\subsection{Electric limits}
\begin{definition}
A semi-relativistic electromagnetic problem is said to be within the \emph{(weak) electric limit}, if 
\begin{equation}\label{lim-ele}
c [B] \lesssim \epsilon [E]
\end{equation}
The problem is said to be within the \emph{strong electric limit}, if
in addition to \eqref{lim-mag}, also
\begin{equation}\label{lim-ele-strong}
[H]\lesssim\epsilon c[D],\qquad
[M]\lesssim\epsilon c[P]
\end{equation}
\end{definition}

Similarly, we have two alternative ways of stating the electric limit:
\begin{align}
c[A]\lesssim\epsilon [\phi],\\ 
[J]\lesssim\epsilon c [\rho].
\end{align}

The reduction of Galilean invariant variables and equations are listed in Table~\ref{tbl:reduction}.
\begin{remark}
The Galilean invariant Poynting's theorem and the momentum identity cannot be reduced in either limit.
\end{remark}

\begin{remark}
In fact, problems within the magnetic or electric limit have some multiscale features. 
For example, within the magnetic limit, the whole set of electric variables are much smaller than
the set of magnetic ones. That is why, as indicated by the reduced formulas of the Galilean invariants,
the former cannot affect the latter, while the latter can affect the former.
Thus, although we have the governing equations containing electric variables that none of the terms in the equations
is negligible, such as \eqref{global-maxwell-2}$_1$, the equations themselves are small as a whole. 
In other words, we do have the formulation of electric phenomena in a magnetic limit problem, but those phenomena
are actually happening in a much smaller scale compare to the scale of the dominating magnetic phenomena.
Only because the difference between two scales is less than $\epsilon^2$, 
the semi-relativistic formulation is capable of capturing both.
\end{remark}


\end{document}